\documentclass[final,5p,times]{elsarticle}

\usepackage{amssymb,amsfonts,amsmath,amsthm}
\usepackage{algorithm}
\usepackage{algpseudocode}

\usepackage{graphicx}
\usepackage[hidelinks]{hyperref}
\usepackage{mathtools}
\usepackage{xparse}

\usepackage{amssymb,amsfonts,amsmath}
\usepackage{graphicx,color,changebar}

\newcommand{\eg}{e.\,g.,\ }

\newcommand{\Sn}{ {\mathbb{S}_n} }   
\usepackage{todonotes}

\newtheorem{remark}{Remark}[section]

\newtheorem{theorem}{Theorem}[section]

\newtheorem{definition}{Definition}[section]

\newcommand{\C}{\mathbb{C}}

\renewcommand{\S}{{\mathcal S}}
\newcommand{\M}{{\mathcal M}}
\newcommand{\V}{{\mathcal V}}
\newcommand{\W}{{\mathcal W}}
\newcommand{\R}{\mathbb{R}}
\newcommand{\XWpd}{ {\mathbb{X}^{\raisebox{0.2em}{{\fontsize{3}{2}\selectfont $>$}}}} }  
\newcommand{\XWpdpd}{ {\mathbb{X}^{\raisebox{0.2em}{{\fontsize{3}{2}\selectfont $\gg$}}}} }  

\newcommand{\diag}[1]{\ensuremath{\mathop{\mathrm{diag}}\left( #1 \right)}}

\newcommand{\ie}{i.\,e.\ }

\usepackage{xstring}
\newcommand{\mat}[3][C]{
	\mathbb{#1}^{
		\IfSubStr{#2}{+}{(#2)}{#2}
		\times
		\IfSubStr{#3}{+}{(#3)}{#3}
	}
}
\DeclareDocumentCommand{\matz}{m m O{K} O{z}}{\mathbb{#3}[{#4}]^
	{
		\IfSubStr{#1}{+}{(#1)}{#1}
		\times
		\IfSubStr{#2}{+}{(#2)}{#2}
	}}

\date{\today}
\usepackage{cleveref}
\usepackage{mathrsfs}
\usepackage{caption}
\usepackage{subcaption}
\usepackage{url}
\usepackage{tikz}
\usepackage{pgfplots}
\newenvironment{customlegend}[1][]{%
    \begingroup
    \csname pgfplots@init@cleared@structures\endcsname
    \pgfplotsset{#1}%
}{  \csname pgfplots@createlegend\endcsname
    \endgroup
}

\def\addlegendimage{\csname pgfplots@addlegendimage\endcsname}

\usetikzlibrary{calc,patterns,
	decorations.pathmorphing,
	decorations.markings}
	
\newlength\fheight
\newlength\fwidth

\journal{***}

\usepackage[software, hardware]{mymacros}

\begin{document}

\begin{frontmatter}
\title{Identification of Port-Hamiltonian Systems from Frequency Response Data}

\author[mainaddress]{Peter Benner}
\ead{benner@mpi-magdeburg.mpg.de}

\author[mainaddress]{Pawan Goyal \corref{mycorrespondingauthor}}
\cortext[mycorrespondingauthor]{Corresponding author. Phone: +49 391 6110 386, Fax: +49 391 6110 453}
\ead{goyalp@mpi-magdeburg.mpg.de}

\author[mainaddress1]{Paul Van Dooren}
\ead{paul.vandooren@uclouvain.be}


\address[mainaddress]{Max Planck Institute for Dynamics of Complex Technical Systems, Sandtorstr. 1, 39106 Magdeburg, Germany}
\address[mainaddress1]{Universit\'e catholique de Louvain, Louvain-La-Neuve, Belgium}



\date{\today}
\address{}
\begin{abstract} 
	In this paper, we study the identification problem of a passive system from tangential interpolation data.  We present a simple  construction approach based on the Mayo-Antoulas generalized realization theory that automatically yields a port-Hamiltonian realization for every strictly passive system with simple spectral zeros. Furthermore, we discuss the construction of a frequency-limited port-Hamiltonian realization. We illustrate the proposed method by means of several examples. 
\end{abstract}


\end{frontmatter}

\section{Introduction} \label{sec:Intro}
We study linear and finite-dimensional dynamical systems that are \emph{passive}. We look at port-Hamiltonian realizations of such transfer functions  
which play an important role in the robustness of passive systems. We consider continuous-time systems that can be represented in the standard 
state-space form with real coefficients and real inputs, outputs and states:
\begin{equation} \label{gstatespace}
 \begin{aligned} \dot x(t) & =  Ax(t) + B u(t),\ \quad x(0)=0,\\
y(t)&= Cx(t)+Du(t).
\end{aligned}
\end{equation}
Denoting real and complex $n$-vectors ($n\times m$ matrices) by $\mathbb R^n$, $\mathbb C^{n}$ ($\mathbb R^{n \times m}$, $\mathbb{C}^{n \times m}$), respectively, 
then $u:\mathbb R\to\mathbb{R}^m$,   $x:\mathbb R\to \mathbb{R}^n$,  and  $y:\mathbb R\to\mathbb{R}^m$  are vector-valued functions denoting the \emph{input}, \emph{state}, 
and \emph{output} of the system, and the coefficient matrices satisfy $A \in \mathbb{R}^{n \times n}$, $B\in \mathbb{R}^{n \times m}$, $C\in \mathbb{R}^{m \times n}$, 
and  $D\in \mathbb{R}^{m \times m}$.
The Hermitian (or conjugate) transpose of a vector or matrix $V$ is denoted by $V^{\mathsf{H}}$ $\left(V^{\mathsf{T}}\right)$ and the identity matrix is 
denoted by $I_n$ or $I$ if the dimension is clear. We furthermore require that input and output dimensions are equal to $m$ since we aim to interpolate with (square) passive transfer functions.  
We denote the set of symmetric matrices in $\mathbb{R}^{n \times n}$ by $\Sn$.
Positive definiteness (semi-definiteness) of  $M\in \Sn$ is denoted by $M\succ 0$ ($M\succeq 0$).  

Model-order reduction for passive systems has been an active research area and has been investigated by several researchers in \eg \cite{morBenQQ04a,morDanP02,morFreF96a,morPhiDS03,gugercin2012structure,polyuga2012effort,wolf2010passivity}. However,  this requires the availability of system matrices, which may not be easily available, especially when the necessary parameters to model a dynamical process are not known. Hence, we aim at identifying system realizations using frequency response. The structure of the paper is as follows. In Section 2, we briefly recall some important properties of passive systems. Subsequently, in Section 3, we discuss state-space representations and properties of port-Hamiltonian realizations. This is followed by a discussion of degrees of freedom of a port-Hamiltonian system in the subsequent section in order to have an understanding how many parameters are needed to describe a minimal port-Hamiltonian system. In Section 5, we propose a variant of the Loewner-based approach, realizing the system in port-Hamiltonian form when data are available at spectral zeros along with zero directions. Furthermore, we discuss the estimation of the dominant spectral zeros and zero directions using the data given on the imaginary axis in Section 6. In Section 7, we illustrate the proposed identification approach by means of a couple of numerical examples, which is followed by a short summary. 

\section{Passive Systems and Port-Hamiltonian Realizations} \label{sec:PH}
\emph{Passive} systems and their relationships with \emph{positive-realness and stability conditions} are well studied.
We briefly recall some important properties by following \cite{Wil72b}, and refer to the literature for a more detailed survey. We consider continuous-time systems with a real
rational transfer matrix $Z(s)$ and define the associated spectral density function: 
\begin{equation} \label{Phi} \Phi(s):=Z^{\mathsf{T}}(-s) + Z(s), 
\end{equation}
which coincides with the Hermitian part of $Z(s)$ on the  $\imath\omega$ axis:
\[ \Phi(\imath \omega)= [Z(\imath \omega)]^\mathsf{H} + Z(\imath \omega). \] 
\begin{definition}
The rational transfer function $Z(s)$ is called {\em strictly positive-real} \/if  $\Phi(\imath\omega)\succ 0$
for all $\omega\in \mathbb{R}$ and it is called \emph{positive-real} if $\Phi(\imath \omega)\succeq 0 $ for all $\ \omega\in \mathbb{R}$.

The transfer function $Z(s)$ is called {\em asymptotically stable} if the poles of the transfer function are in the open left half-plane, and it is called 
{\em stable} if all the poles are in the closed left half-plane, with any pole occurring on the imaginary axis being first-order.

The transfer function $Z(s)$ is called {\em strictly passive} if it is strictly positive-real and asymptotically stable and it is called \emph{passive} if it is positive-real and stable. 
\end{definition}

We will assume throughout the paper that this realization is minimal (\ie controllable and observable) and will restrict ourselves in this paper to {\em strictly} passive systems, which implies that the matrix $A$ is invertible 
and that the transfer matrix is proper since  poles cannot be on the imaginary axis or at infinity. 
Moreover, $\Phi(\imath\omega)\succ 0$ at $\omega=\infty$, implies that we must have $D^{\mathsf{T}}+D\succ 0$ as well. 
We will see that this restriction simplifies our discussion significantly. 
It is also a reasonable restriction because passive systems can be viewed as limiting cases of strictly passive systems.

Since the transfer function is proper, we can represent it in standard state-space form $Z(s)=C(sI_n-A)^{-1}B+D$. For proper transfer functions $Z(s)$ with minimal realization $\M:=\{A,B,C,D\}$, there is a necessary and sufficient condition for passivity, known 
as the Kalman-Yakubovich-Popov linear matrix inequality. An elegant proof of this can be found in \cite{Wil72b}. 
\begin{theorem}
Let  $\M:=\left\{A,B,C,D\right\}$ be a minimal realization of a proper rational transfer function $Z(s)$ and let 
\begin{equation} \label{LMI}
 W(X,\M)=\begin{bmatrix} -A^\mathsf{T} X - XA & C^\mathsf{T}- XB \\ C-B^\mathsf{T}X & D+D^\mathsf{T} \end{bmatrix}.
\end{equation}
Then, $Z(s)$ is passive if and only if there exists a real symmetric matrix $X\in \Sn$ such that
\begin{equation} \label{KYP1}    W(X,\M) \succeq 0 , \quad X\succ 0,
\end{equation}
and is  strictly passive if and only if there exists a real symmetric matrix $X\in \Sn$ such that
\begin{equation} \label{KYP2}    W(X,\M) \succ 0 , \quad X \succ 0.
\end{equation}
\end{theorem}
The solutions $X$ of these inequalities are known as {\em certificates} for the passivity or strict passivity of the system $\M$. 
\begin{definition}
Every solution $X$ of the LMI 
\begin{equation}
\XWpd :=\left\{ X\in \S \left|   W(X,\M) \succeq 0,\ X \succ 0 \right.\right\} \label{XpdsolWpsd} 
\end{equation}
is called a {\em certificate} for the passivity of the model $\M$ and every solution of the LMI 
\begin{equation}
\XWpdpd :=\left\{ X\in \S \left|   W(X,\M) \succ 0,\ X \succ 0 \right.\right\} \label{XpdsolWpd}
\end{equation}
is called a {\em certificate} for the strict passivity of the model $\M$.
\end{definition}

\section{Port-Hamiltonian Models}
In this section, we  provide a brief introduction to special realizations of passive systems, known as port-Hamiltonian system models.
\begin{definition}\label{def:ph}
A linear time-invariant \emph{port-Hamiltonian system model} of a proper transfer function, has the standard state-space form
\begin{equation} \label{pH}
 \begin{aligned} \dot x(t)  & =  (J-R)Q x(t) + (G-P) u(t),\quad x(0) =0,\\
y(t)&= (G+P)^{\mathsf{T}}Qx(t)+(N+S)u(t),
\end{aligned}
\end{equation}
where the system matrices 
\begin{equation} \label{sym}
\V:= \begin{bmatrix} J & G \\ -G^{\mathsf{T}} & N \end{bmatrix}, \quad
\W:= \begin{bmatrix} R & P \\ P^{\mathsf{T}} & S \end{bmatrix}, 
\end{equation}
satisfy the conditions
$$ \mathcal V =-\mathcal V^{\mathsf{T}}, \quad \mathcal W =\mathcal W^{\mathsf{T}}\succeq 0, \quad Q=Q^{\mathsf{T}} \succeq 0.
$$
\end{definition}
It readily follows from the properties of port-Hamiltonian  models that when $Q$ and $\mathcal W$ are invertible, we can choose $X=Q$ as certificate
for the model $$\M:=\left\{(J-R)Q,G-P,(G+P)^\mathsf{T}Q,N+S\right\}$$ to show that it satisfies the strict passivity condition \eqref{KYP2}. 
\begin{remark}
The condition that $Q$ is non-singular is automatically satisfied when the state transition matrix $A$ is non-singular,
which is the case for strictly passive systems. We can then also represent the system in generalized state-space form,
using $\widehat x:=Qx$, yielding:
\begin{equation} \label{gpH}
 \begin{aligned} Q^{-1} \dot {\widehat x}  & =  (J-R) \widehat x + (G-P) u,\\
y&= (G+P)^{\mathsf{T}} \widehat x+(N+S)u.
\end{aligned}
\end{equation}
We use such models for representing intermediate results later on.
\end{remark}

Conversely, let  $\M:=\left\{A,B,C,D\right\}$ be a state-space model, satisfying the strict passivity condition \eqref{KYP2} with a given certificate $X\succ 0$.
Then, it can always be transformed in the port-Hamiltonian form, as shown in \cite{BMX16}.
We can use a symmetric factorization $X= T^{\mathsf{T}}T$, which implies the invertibility of $T$, and define a new realization 
\[
\{A_T,B_T,C_T,D\} := \{TAT^{-1}, TB, CT^{-1}, D \}
\]
so that
$$
\begin{bmatrix} T^{-\mathsf{T}} & 0\\ 0 & I_m
\end{bmatrix}
\begin{bmatrix} -A^{\mathsf{T}}X-XA & C^{\mathsf{T}}-XB \\ C-B^{\mathsf{T}}X & D^{\mathsf{T}}+D
\end{bmatrix}
\begin{bmatrix} T^{-1} & 0\\ 0 & I_m
\end{bmatrix}
 $$
\begin{equation}
\label{PH}
= 
\begin{bmatrix}-A_T & -B_T \\ C_T & D
\end{bmatrix}+
\begin{bmatrix} -A_T^{\mathsf{T}} & C^{\mathsf{T}}_T \\ -B^{\mathsf{T}}_T & D^{\mathsf{T}}
\end{bmatrix}
\succ  0.
\end{equation}
We can then use the symmetric and skew-symmetric part of the matrix
\[
 \S := \begin{bmatrix}-A_T & -B_T \\ C_T & D
\end{bmatrix}
\]
to define the coefficients of a port-Hamiltonian  representation  via
\[
\V := \begin{bmatrix} J &  G \\ -G^{\mathsf{T}} & N \end{bmatrix} :=  \frac{{\S} - {\S}^{\mathsf{T}}}{2},  \;
\W := \begin{bmatrix} R & P \\ P^{\mathsf{T}} & S \end{bmatrix} := \frac{{\S} +{\S}^{\mathsf{T}}}{2}.
\]
This construction yields $\mathcal W \succ 0$ and $Q=I_n$ because of the chosen factorization $X=T^{\mathsf{T}}T$. This is called a {\em normalized}
port-Hamiltonian representation. This shows that port-Hamiltonian models with strict inequalities $Q\succ 0$ and $\mathcal W \succ 0$ 
are nothing but strictly passive systems described in an appropriate coordinate system.
On the other hand, it was shown in \cite{MeV19} that normalized port-Hamiltonian systems have good robustness properties
in terms of their so-called passivity radius.

\section{Degrees of Freedom of a Transfer Function} \label{sec:DOF}

In the literature, one can find a discussion of the degrees of freedom of a given strictly proper rational transfer function~$Z(s)$ of a given MacMillan degree $n$ \cite{HaK75}.
This corresponds to the minimum number of parameters to describe such a function. Since this literature is quite opaque, we briefly re-derive the basic results using a generic 
$m\times m$ strictly proper transfer matrix of MacMillan degree $n$ without repeated poles. Such a transfer function can be written in its partial fraction expansion as follows:
$$ 
Z(s) = \sum_{k=1}^{n_r} u_k(s-\lambda_k)^{-1}v_k^\mathsf{T}+ 
\sum_{k=1}^{(n-n_r)/2} U_k\left(sI_2-\begin{bmatrix}\alpha_k & \beta_k \\ -\beta_k & \alpha_k \end{bmatrix}\right)^{-1}V_k^\mathsf{T},
$$
which requires a total of $2(m+1)n$ real parameters.  This can be seen as a state-space model in the real Jordan form with $1\times 1$ diagonal elements for the $n_r$ real poles 
and $2\times 2$ diagonal blocks for the $n_c:=n-n_r$ complex conjugate complex poles.
But this representation is only unique up to a block diagonal state-space transformation with exactly $m$ degrees of freedom: a scalar $t_k$ for each real pole and a 
$2\times 2$ block $t_k\begin{bmatrix}c_k & s_k \\ -s_k & c_k \end{bmatrix}$ for each complex conjugate pair, where the real rotation matrix depends on one real parameter. 
When taking the quotient of the 
manifold of block-diagonal models with respect to this state space transformation, we are left with the exact number of real degrees of freedom, which is
$2mn$ for a strictly proper $m\times m$ transfer function of degree $n$ with real coefficients. 

When considering the larger class of real $m\times m$ {\em proper} rational transfer functions, one has to add the real parameters to {\em realize} the constant matrix $D$. 
If $D$ is constrained to have a particular rank, then we again need to take that into account. A rank $r$ matrix $D$ can be represented
by a rank factorization $D=UV^\mathsf{T}$ where we can again quotient out the degrees of freedom of an $r\times r$ factor $T$ in an equivalent
factorization $D=(UT)(T^{-1}V^\mathsf{T})$. Such a factor can thus be represented by $r(2m-r)$ degrees of freedom, which has to be added to those of 
the strictly proper part of $Z(s)$.  

To summarize, a real rational $m\times m$ transfer function $Z(s)$ of MacMillan degree $n$ has 
\begin{itemize}
 \item  $2mn$ real degrees of freedom when $Z(s)$ is strictly proper,
 \item  $2m(n+r) -r^2$  real degrees of freedom when $Z(s)$ is proper and $Z(\infty)$ has rank $r$.
\end{itemize}

This count of the number of degrees of freedom will determine the number of parameters we can assign using tangential 
interpolation conditions. For a rigorous discussion on these aspects, we refer to \cite{HaK75}.

\section{Loewner Approach for Identification of a port-Hamiltonian  Realization}

In this section, we discuss the identification problem of a strictly passive transfer function $Z(s)$ of degree $n$, which is defined via a set of 
left and right interpolation conditions. Since $Z(s)$ is strictly passive, it is proper and has 
a standard state-space realization $\{A,B,C,D\}$ with $D$ of full rank and positive-real (\ie $D+D^\mathsf{T}\succ 0$).
We can then define the transfer function $Z(s)$ via a set of left and right tangential interpolation conditions 
\begin{equation} \label{interpol}
v_j:= \ell_jZ(\mu_j), ~~   w_j:=Z(\lambda_j)r_j, ~~  j=1,\ldots,n, ~~ Z(\infty)= D,
\end{equation}
where $(\mu_j,\ell_j,v_j)$, and $(\lambda_j,r_j,w_j), \; j=1,\ldots, n$, are sets of self-conjugate left and right interpolation conditions with \linebreak
$\{\ell_j,v_j\}\in \C^{1\times m}$, $\{r_j,w_j\}\in \C^{m\times 1}$, $\{\lambda_j,\mu_j\}\in\C $. 
Then, the so-called {\em Loewner} and {\em shifted Loewner} matrices defined in \cite{Ant05}  have dimensions $n\times n$ and look like

\begin{subequations}\label{MaAnt} 
	\begin{align}
	\mathbb{L}& := \begin{bmatrix}
	\frac{\ell_1w_1-v_1r_1}{\lambda_1-\mu_1} & \ldots & \frac{\ell_1w_n-v_1r_n}{\lambda_n-\mu_1} \\ 
	\vdots & \ddots & \vdots \\ 
	\frac{\ell_nw_1-v_nr_1}{\lambda_1-\mu_n} & \ldots & \frac{\ell_nw_n-v_nr_n}{\lambda_n-\mu_n} \end{bmatrix}\!,\\
 \mathbb{L}_\sigma &:= \begin{bmatrix}
\frac{\lambda_1\ell_1w_1-\mu_1v_1r_1}{\lambda_1-\mu_1} & \ldots & \frac{\lambda_n\ell_1w_n-\mu_1v_1r_n}{\lambda_n-\mu_1} \\
\vdots & \ddots & \vdots \\ 
\frac{\lambda_1\ell_nw_1-\mu_nv_nr_1}{\lambda_1-\mu_n} & \ldots & \frac{\lambda_n\ell_nw_n-\mu_nv_nr_n}{\lambda_n-\mu_n}                       
\end{bmatrix}.
	\end{align}
\end{subequations}
They satisfy the following Sylvester equations 
\begin{equation} \label{Lyaplike}  \mathbb{L}\Lambda - M \mathbb{L}= LW-VR , \quad
\mathbb{L}_\sigma\Lambda - M \mathbb{L}_\sigma = LW\Lambda-M VR,  
\end{equation} 
where we used the definitions
\begin{equation} \label{LV} L := \begin{bmatrix} \ell_1 \\ \vdots \\ \ell_n \end{bmatrix}, \quad  V := \begin{bmatrix} v_1 \\ \vdots \\ v_n \end{bmatrix}, \quad
   M := \diag{\mu_1, \ldots, \mu_n},                                                        
\end{equation} 
and
\begin{align*}
R:=\begin{bmatrix} r_1,\ldots,r_n \end{bmatrix}, \quad W:=\begin{bmatrix} w_1,\ldots,w_n \end{bmatrix},
\end{align*}
\begin{equation} \label{RW}   \Lambda := \diag{\lambda_1, \ldots, \lambda_n}.                                                         
\end{equation} 
The following interpolation result follows from the theory developed in \cite{MaA07} in the special case that the Loewner matrix $\mathbb{L}$ is invertible.
\begin{theorem} \label{LoewnerIdent} 
Let $Z(s)$ be a proper transfer function of MacMillan degree $n$, then the interpolation conditions \eqref{interpol} uniquely define $Z(s)$ if
the Loewner matrix $\mathbb{L}$ is invertible. Moreover, a minimal generalized state-space realization is then given by
$$ Z(s)=(W-DR)(\mathbb{L}_\sigma - LDR- s\mathbb{L})^{-1}(V-LD)+D$$
and the corresponding system matrix is given by
$$ \left[ \begin{array}{c|c}  A -s E &  B \\ \hline  C &  D \end{array}\right] = 
\left[ \begin{array}{c|c} \mathbb{L}_\sigma - s \mathbb{L} & V \\ \hline -W & 0 \end{array}\right]
+ 
\begin{bmatrix}- L \\ \hline I_m \end{bmatrix}  D \left[ \begin{array}{c|c} R & I_m \end{array}\right].$$

\end{theorem}

\medskip

Let us now apply this to the special case where the interpolation points are the so-called spectral zeros of $Z(s)$.

\begin{definition}
 Let $Z(s)$ be a real and strictly passive transfer function of MacMillan degree $n$ with associated spectral density function $\Phi(s):= Z^\mathsf{T}(-s) + Z(s)$. 
 Then the spectral zeros and zero directions of $Z(s)$ are the pairs $(s_j,r_j)$ such that
 $\Phi(s_j)r_j=0$. 
\end{definition}
When the zeros are distinct (which is generic), there are exactly $n$ zeros in the open right half-plane and $n$ zeros in the open left half-plane because the 
spectral density function $\Phi(s)$ has degree $2n$ and has no zeros on the imaginary axis. The definition of the spectral zeros implies
$$  \Phi(s_j)r_j=Z^\mathsf{T}(-s_j)r_j+Z(s_j)r_j=0 , 
$$
and hence
$$  w_j := Z(s_j)r_j \quad \Longleftrightarrow \quad Z^\mathsf{T}(-s_j)r_j=-w_j.
$$
Since the spectral zeros and zero directions $(s_j,r_j)$ form a self-conjugate set, we can distinguish two cases for these equations,
depending on the condition that $s_j$ is a real zero or not. In the real case, we have
$$ s_j\in \R : \quad  Z(s_j)r_j=w_j  \quad \Longleftrightarrow  \quad r_j^\mathsf{T}Z(-s_j)= -w_j^\mathsf{T},
$$
and in the complex case, we have
$$ s_j\notin \R : \quad
\left\{ \begin{array}{c} Z(s_j)r_j=w_j \quad \Longleftrightarrow  \quad r_j^\mathsf{H}Z(-\overline s_j)= -w_j^\mathsf{H}, \\
         Z(\overline s_j)\overline r_j=\overline w_j  \quad \Longleftrightarrow  \quad r_j^\mathsf{T}Z(-s_j)= -w_j^\mathsf{T}.
        \end{array} 
\right.
$$
Therefore, if we define $\lambda_j,\; j=1,\ldots,n$, to be the spectral zeros of $Z(s)$ in the open right half-plane, 
$$ \Re \lambda_j \ge 0, \quad Z(\lambda_j)r_j=w_j,  \; j=1,\ldots,n, $$
then the set of right tangential conditions $(\lambda_j,r_j,w_j)$ is self-conjugate. Moreover,
to every right tangential condition \linebreak 
$Z(\lambda_j)r_j=w_j$ (and its complex conjugate when $\lambda_j$ is complex), there is a corresponding left tangential condition
$$ r_j^\mathsf{H}Z(-\overline \lambda_j)= -w_j^\mathsf{H}, \; j=1,\ldots,n. $$
Therefore, we can define left tangential interpolation conditions $\ell_j Z(\mu_j)=v_j, \; j=1,\ldots ,n$ in
such a way that 
$$  M=-\overline \Lambda=-\Lambda^\mathsf{H}, \qquad L=  R^\mathsf{H}, \qquad V= -W^\mathsf{H}.
$$
Using these definitions, the Loewner and shifted Loewner matrices now become
\begin{subequations}\label{LLs}
	\begin{align}
	\mathbb{L} &:= \begin{bmatrix}
	\frac{r_1^\mathsf{H}w_1+w_1^\mathsf{H}r_1}{\lambda_1+\overline \lambda_1} & \ldots & \frac{r_1^\mathsf{H}w_n+w_1^\mathsf{H}r_n}{\lambda_n+\overline \lambda_1} \\
	\vdots & \ddots & \vdots \\ 
	\frac{r_n^\mathsf{H}w_1+w_n^\mathsf{H}r_1}{\lambda_1+\overline \lambda_n} & \ldots 
	& \frac{r_n^\mathsf{H}w_n+w_n^\mathsf{H}r_n}{\lambda_n+\overline \lambda_n} \end{bmatrix}\!,\\
 \mathbb{L}_\sigma &:= \begin{bmatrix}
\frac{\lambda_1r_1^\mathsf{H}w_1-\overline \lambda_1w_1^\mathsf{H}r_1}{\lambda_1+\overline \lambda_1} & 
\ldots & \frac{\lambda_nr_1^\mathsf{H}w_n-\overline \lambda_1w_1^\mathsf{H}r_n}{\lambda_n+\overline \lambda_1} \\
\vdots & \ddots & \vdots \\ 
\frac{\lambda_1r_n^\mathsf{H}w_1-\overline \lambda_nw_n^\mathsf{H}r_1}{\lambda_1+\overline \lambda_n} & \ldots &
\frac{\lambda_nr_n^\mathsf{H}w_n-\overline \lambda_nw_n^\mathsf{H}r_n}{\lambda_n+\overline \lambda_n}                       
\end{bmatrix}
	\end{align}
\end{subequations}
and they satisfy the equations 
\begin{equation} \label{Lyapsym}  \mathbb{L}\Lambda + \Lambda^\mathsf{H} \mathbb{L}= R^\mathsf{H}W+W^\mathsf{H}R , \quad
\mathbb{L}_\sigma\Lambda + \Lambda^\mathsf{H} \mathbb{L}_\sigma = R^\mathsf{H}W\Lambda-\Lambda^\mathsf{H}W^\mathsf{H}R. 
\end{equation} 
We point out that the matrix $\mathbb{L}$ is Hermitian by the construction itself, while $\mathbb{L}_\sigma$ is skew-Hermitian by construction.
For such symmetric conditions, the matrix $\mathbb{L}$ is also called the Pick matrix (see \cite{YoS67,morAnt05a}).   
It follows from Theorem \ref{LoewnerIdent} that a generalized state-space realization $\{A,B,C,D,E\}$ is given by
\begin{equation}
 \label{Loewner} 
\left[ \begin{array}{c|c}  A -s  E &  B \\ \hline  C &  D \end{array}\right] = \left[ \begin{array}{c|c} \mathbb{L}_\sigma - s \mathbb{L} & -W^\mathsf{H} \\ \hline -W & 0 \end{array}\right]
+ 
\begin{bmatrix}- R^\mathsf{H} \\ \hline I_m \end{bmatrix}  D \left[ \begin{array}{c|c} R & I_m \end{array}\right].
\end{equation}

Notice that the introduction of complex matrices and vectors in this section is in fact artificial. Since the interpolation conditions are self-conjugate, 
we can transform the construction as follows. Let $v:=v_r+\imath v_\imath$ be a complex vector associated with a complex interpolation point 
$\lambda := \alpha+\imath\beta$, then the unitary transformation 
$Q:=\frac{1}{\sqrt{2}}\begin{bmatrix}1 & -\imath \\ 1 & \imath \end{bmatrix}$ transforms pairs of complex conjugate data to real data, as can be seen below
$$ \begin{bmatrix} v & \overline v  \end{bmatrix} \cdot Q = \sqrt{2} \begin{bmatrix} v_r & v_\imath  \end{bmatrix}, \quad 
Q^\mathsf{H} \cdot \begin{bmatrix} \lambda & 0 \\ 0 & \overline \lambda \end{bmatrix} \cdot Q = \begin{bmatrix}  \alpha &  \beta \\ - \beta  & \alpha \end{bmatrix}.
$$ 
If the pairs of complex conjugate vectors and interpolation points have been permuted to be adjacent, then it suffices to
apply a block diagonal unitary similarity transformation $U$ with diagonal blocks $Q$ corresponding to each complex conjugate pair $(\lambda,\overline \lambda)$,
to transform the equations \cref{LLs,Lyapsym,RW} to real equations:
\begin{subequations} \label{RealLyap}  
\begin{align}
\mathbb{\widehat L} \Omega +\Omega^\mathsf{T}\mathbb{\widehat L}&= \widehat R^\mathsf{T}\widehat W+\widehat W^\mathsf{T}\widehat R , ~~\text{and}\\
 \mathbb{\widehat L}_\sigma\Omega + \Omega^\mathsf{T} \mathbb{\widehat L}_\sigma &= \widehat R^\mathsf{T}\widehat W \Omega-\Omega^\mathsf{T}\widehat W^\mathsf{T}\widehat R,
\end{align}
\end{subequations} 
where
$$ \widehat{\mathbb{L}}=U^\mathsf{H}\mathbb{L}U, \; \mathbb{\widehat L}_\sigma=U^\mathsf{H}\mathbb{L}_\sigma U, \; \Omega=U^\mathsf{H}\Lambda U, 
 \; \widehat W=W U,  \; \widehat R= R U, $$
and $\Omega$ is now block diagonal with $2\times 2$ blocks corresponding to each pair of complex conjugate interpolation points.
It then also follows from \eqref{Loewner} that a {\em real} generalized state-space realization $\{\widehat A,\widehat B,\widehat C, D,\widehat E\}$ is then given by
\begin{equation}
 \label{RealLoewner} 
\left[ \begin{array}{c|c} \widehat A -s\widehat E & \widehat B \\ \hline \widehat C & D \end{array}\right] = 
\left[ \begin{array}{c|c} \mathbb{\widehat L}_\sigma - s \mathbb{\widehat L} & -\widehat W^\mathsf{T} \\ \hline -\widehat W & 0 \end{array}\right] 
+ 
\begin{bmatrix}- \widehat R^\mathsf{T} \\ \hline I_m \end{bmatrix}  D \left[ \begin{array}{c|c} \widehat R & I_m \end{array}\right].
\end{equation}

Let us now look at the passivity condition we imposed on the transfer function $Z(s)$. The Loewner matrix $\mathbb{L}$ given in
\eqref{LLs} has the structure of a Pick matrix (see \eg \cite{YoS67}) since the spectral zeros used for the interpolation are assumed to be distinct.
The strict passivity of $Z(s)$ implies that this matrix is positive definite. It follows that $Z(\infty)=D$, and hence that
$D$ must be strictly positive-real as well. Since $\mathbb{L}$ is positive definite, so is $\mathbb{\widehat L}$ and we can factorize it as $\mathbb{\widehat L}=\Gamma^\mathsf{T}\Gamma$,
where $\Gamma$ is invertible, by using, for instance, the upper triangular Cholesky factor. Defining 
$$  \widehat W_\Gamma:=\widehat W \Gamma^{-1}, \quad \widehat R_\Gamma:=\widehat R \Gamma^{-1}, \quad \mathbb{\widehat L}_{\sigma \Gamma}:= \Gamma^{-\mathsf{T}}\mathbb{\widehat L}_{\sigma}\Gamma^{-1},
$$
we obtain an equivalent quadruple for the state-space realization 
$\{\widehat A_\Gamma,\widehat B_\Gamma,\widehat C_\Gamma, D\}=\{\Gamma^{-\mathsf{T}}\widehat A\Gamma^{-1},\Gamma^{-\mathsf{T}}\widehat B,\widehat C\Gamma^{-1}, D\}$ of $Z(s)$ as
\begin{equation}
 \label{Loewner2} 
\left[ \begin{array}{c|c} \widehat A_\Gamma & \widehat B_\Gamma \\ \hline \widehat C_\Gamma & D \end{array}\right] = 
\left[ \begin{array}{c|c} \mathbb{\widehat L}_{\sigma \Gamma} & -\widehat W_\Gamma^\mathsf{T} \\ \hline -\widehat W_\Gamma & 0 \end{array}\right] +
\begin{bmatrix}- \widehat R_\Gamma^\mathsf{T} \\ \hline I_m \end{bmatrix} D \left[ \begin{array}{c|c} \widehat R_\Gamma & I_m \end{array}\right].
\end{equation}
We then show  that this realization is in port-Hamiltonian form.
\begin{theorem}
 Let us construct an $m\times m$ real transfer function $Z(s)$ of MacMillan degree $n$ using the self-conjugate interpolation conditions as follows:
 $$ Z(\infty) =  D, \quad  Z(\lambda_j)r_j=w_j, \quad r^\mathsf{H}_j  Z(-\overline \lambda_j)=-w_j^\mathsf{H}, \quad j=1,\ldots, n, $$
where $\Re\left(\lambda_j\right) > 0$,  $D+D^\mathsf{T}\succ 0$ and $\mathbb{\widehat L}\succ 0 $ in which $\Re\left(\cdot\right)$ denotes the real part. Then, $Z(s)$ is strictly passive and the quadruple $\{\widehat A_\Gamma,\widehat B_\Gamma,\widehat C_\Gamma, D \}$ 
is in normalized port-Hamiltonian form and its spectral zeros and zero directions are given by $(\lambda_j,r_j),\; j=1,\ldots,n$.
\end{theorem}
\begin{proof}
 A necessary condition for strict passivity is that the Hermitian part of $Z(s)$ is positive definite on the whole imaginary axis, including infinity,  and since $D=Z(\infty)$ and is a real matrix, we must have $D+D^\mathsf{T}\succ 0$. A necessary and sufficient condition for the passivity  of $Z(s)$ with given interpolation data, is that the Loewner matrix $\mathbb{\widehat L}$ is positive semi-definite, but since we assume $\mathbb{\widehat L}\succ 0$,  the transfer function is passive. Let us now decompose the real matrix $D$ as $D =  N +  S$, where $S$ is the symmetric part of $D$ and $N$ is its skew-symmetric part. Then, following the discussion of Section \ref{sec:PH}, we obtain
 $$ \W =  \W^\mathsf{T} =
\left[\begin{array}{c} \widehat R_\Gamma^\mathsf{T} \\ \hline I_m \end{array}\right]  S \left[ \begin{array}{c|c} \widehat R_\Gamma & I_m \end{array}\right] \succeq 0,$$
$$ \V = - \V^\mathsf{T} = \left[ \begin{array}{c|c} - \mathbb{\widehat L}_{\sigma \Gamma} & \widehat W_\Gamma^\mathsf{T} \\ \hline -\widehat W_\Gamma & 0 \end{array}\right] \! + \!
\left[\begin{array}{c} \widehat R_\Gamma^\mathsf{T} \\ \hline I_m \end{array}\right]  N \left[ \begin{array}{c|c} \widehat R_\Gamma & I_m \end{array}\right]
 $$
 which are the conditions for the passivity of a normalized port-Hamiltonian system.  The standard state-space realization \eqref{Loewner2} 
 is therefore normalized port-Hamiltonian. It follows from the self-conjugate conditions that 
 $$   \Phi^\mathsf{T}(-\lambda_j)r_j=\Phi(\lambda_j)r_j=Z(-\lambda_j)r_j+Z(\lambda_j)r_j=-w_j+w_j=0, 
 $$
 for $j=1,\ldots,n,$ and since $\Phi(s)$ has MacMillan degree bounded by $2n$, these are the only zeros of $\Phi(s)$, which also implies that $Z(s)$ is then strictly passive.
\end{proof}

\begin{remark}
The conditions that the spectral zeros should be simple can be removed. The construction of the Loewner matrix $\mathbb{L}$ and of the shifted Loewner matrix $\mathbb{L}_\sigma$
then have to be adapted, as explained in \cite{morAnt05a,MaA07}, but the properties of these matrices are preserved.
The tangential interpolation conditions then also involve tangential conditions on the derivatives of $Z(s)$ at the spectral zeros $\lambda_i$, but the conclusions remain the same.
\end{remark}

\begin{remark}
The conditions that we should know the zero directions of the corresponding spectral zeros of $Z(s)$ form a demanding constraint. 
But this is different in the scalar case since we only need to impose a scalar condition $Z(-\lambda_j)+Z(\lambda_j)$ in each spectral zero. 
We can then choose $R=\begin{bmatrix} 1, \ldots, 1 \end{bmatrix}$ which implies that $W= \begin{bmatrix} Z(\lambda_1), \ldots, Z(\lambda_n) \end{bmatrix}$.
\end{remark}

Finally, we summarize the construction of a port-Hamiltonian realization in the normalized form in \Cref{Algo:Construct_pH_realization}.
\begin{algorithm}[tb]
	\caption{Construction of a port-Hamiltonian realization in a normalized form.}\label{Algo:Construct_pH_realization}
	\textbf{Input:} \vspace{-0.25cm}
	\begin{itemize} 
	  \setlength{\itemsep}{1pt}
	  \setlength{\parskip}{0pt}
	  \setlength{\parsep}{0pt}
	  \item Spectral zeros $\lambda_j$ and zero directions $r_j$, $j = 1,\ldots,n$, 
	  \item transfer function measurements, i.e. $w_j = Z(\lambda_j)r_j$, where $Z(s)$ denotes the transfer function, 
	  \item the feed-through term $D$. 
	\end{itemize}\vspace{-0.25cm}
	\begin{algorithmic}[1]
		\State Construct the Loewner and shifted Loewner matrices using $w_j$ and $r_j$ as shown in \eqref{LLs}.
		\State Define $W := \begin{bmatrix}w_1,\ldots,w_n \end{bmatrix}$ and $R := \begin{bmatrix}r_1,\ldots,r_n \end{bmatrix}$.
		\State Construct the interpolating realization, ensuring the matching of the transfer function at infinity is:
		\Statex $E =\mathbb L$, $A = \mathbb L_s - R^{\mathsf{H}}DR$, $B = -W^{\mathsf{H}} - R^{\mathsf{H}}D$, $C = -W + DR$.
		\State Perform the unitary transformation to obtain a real realization $(\widehat E,\widehat A,\widehat B,\widehat C)$.
		\State Consider the Choleskey factorization of $\widehat E:= \Gamma^{\mathsf{T}}\Gamma$. 
		\State Construct a port-Hamiltonian realization in the normalized form as follows:
		\Statex $\tA = \Gamma^{-\mathsf{T}}\widehat A\Gamma^{-1}$,\quad $\tB = \Gamma^{-\mathsf{T}}\widehat B$, \quad $\tC = \widehat C\Gamma^{-1}$.
	\end{algorithmic}
	\textbf{Output:} A port-Hamiltonian realization:~$\left(\tA,\tB,\tC,D\right)$. 
\end{algorithm}

\section{Estimation of Spectral Zeros and Zero Directions using Data on the Imaginary Axis}
So far, we have discussed how to construct a port-Hamiltonian realization from the transfer function measurements on spectral zeros along with zero directions. However, it may be restrictive as in practice, it is almost impossible to know the spectral zeros and zero directions a priori. Moreover, even if the zeros are known, taking measurements at those points and directions is not straightforward. On the other hand, there are methods allowing us to obtain  measurements on the imaginary axis. Using these measurements, one can obtain a realization using the classical Loewner approach, proposed in \cite{MaA07}, which interpolates the data. However, it is very likely that it will not yield a realization in normalized port-Hamiltonian form. But we are interested in a passive realization given there is an underlying passive system. To do so, we first propose a strategy in Algorithm~\ref{Algo:EstimateSpectralZero} to estimate the spectral zeros and directions based on the data on the $j\omega$ axis. Once we have such a data set, we can obtain a passive realization directly using \Cref{Algo:Construct_pH_realization}.

\begin{algorithm}[tb]
	\caption{Estimation of spectral zeros and directions using the  data on the imaginary axis.}\label{Algo:EstimateSpectralZero}
	\begin{algorithmic}[1]
		\State Collect enough samples on the imaginary axis. 
		\State Construct a realization using the Loewner approach, ensuring the matching of the transfer function at infinity. 
		\State Determine spectral zeros of the obtained realization. 
	\end{algorithmic}
\end{algorithm}

The main motivation of proposing \Cref{Algo:EstimateSpectralZero} is as follows. As we know, if the transfer functions of two linear systems are the same, then there exists a state-space transformation, allowing us to go from one to another. Furthermore, it is also known that a minimal realization of a linear system can be obtained using the Loewner approach, assuming we have enough samples on the imaginary axis. Hence, if there exists a passive realization of the linear system, then such a passive realization can be determined using the realization obtained using the Loewner approach and a state-space transformation. However, a state-space transformation of a linear system does not change the spectral zeros and corresponding directions. Consequently, we can indeed directly use the realization obtained using the Loewner approach to estimate the spectral zeros and corresponding directions and further can evaluate the transfer function at spectral zeros and in the corresponding directions. 
 
 \begin{remark}
  One can also construct a reduced-order system as well by truncating singular values of the Loewner matrix at a desired tolerance. This can be followed by determining the spectral zeros and zero directions of the reduced-order system, which can be very different from the original ones; however, the spectral zeros and zero directions of the reduced-order system form a good representative, allowing us to compare the important dynamics of the original system. 
 \end{remark}
 \begin{remark}
  If the transfer function measurements are given in a particular frequency band, then applying \Cref{Algo:EstimateSpectralZero} would yield  spectral zeros and zero directions, corresponding to the considered frequency band. If a port-Hamiltonian realization in the normalized form is constructed using \Cref{Algo:Construct_pH_realization}, then we obtain the frequency-limited port-Hamiltonian realization. This is discussed and illustrated further in the subsequent section.  
 \end{remark}

\section{Illustrative Examples and Application in Model-Order Reduction}
In this section, we illustrate the proposed identification approach to construct a passive realization by means of several examples. All numerical simulations are carried out in MATLAB\textsuperscript{\textregistered} version 7.11.0.584 (R2016b) 64-bit on an \intel \coreiseven-6700 CPU @ 3.40GHz, 6MB cache, 8GB RAM, Ubuntu 16.04.6 LTS (x86-64). 
\subsection{An analytical example}
We first consider an analytical example, showing the necessary steps, precisely \Cref{Algo:Construct_pH_realization}, to identify an underlying passive realization whose transfer function is as follows:
$$ Z(s) := dI_2 - (sI_2 - A)^{-1} \quad \mathrm{with} \quad 
A:=\begin{bmatrix} a & b \\ -b & a \end{bmatrix}$$
and let us take $a=-1$, $b=1$, $d=2$ to make the system strictly passive since it is then port-Hamiltonian with positive definite matrix $\W$.
The poles of $Z(s)$ are the eigenvalues of $A$ and are equal to $-1\pm \imath$ and hence asymptotically stable. The spectral zeros are the zeros of
$\Phi(s)= Z^\mathsf{T}(-s)+Z(s)$ which can be obtained from
\begin{align*}
 Q\Phi(s)Q^\mathsf{H} &= QZ^\mathsf{T}(-s)Q^\mathsf{H}+QZ(s)Q^\mathsf{H} \\
 & = 2dI_2 - (-sI_2 - QA^\mathsf{T}Q^\mathsf{H})^{-1} - (sI_2 - QAQ^\mathsf{H})^{-1},
\end{align*}
where 
\begin{align*}
 Q &=\frac{1}{\sqrt{2}}\begin{bmatrix} 1 & -\imath \\ 1 & \imath \end{bmatrix}, 
 &QAQ^\mathsf{H}&= \begin{bmatrix} a+\imath b & 0 \\ 0 & a-\imath b \end{bmatrix}.
\end{align*}
It then turns out that both $QZ(s)Q^\mathsf{H}$ and $Q\Phi(s)Q^\mathsf{H}$ are diagonal and equal to 
\begin{align*}
QZ(s)Q^\mathsf{H} &= \diag{2-\frac{1}{s+1-\imath},2-\frac{1}{s+1+\imath}}, \\
Q\Phi(s)Q^\mathsf{H} &= \diag{\frac{6+8\imath s-4s^2}{2+2\imath s-s^2},\frac{6-8\imath s-4s^2}{2-2\imath s-s^2}}.
\end{align*}
The spectral zeros in the right half-plane are $\lambda=\frac{\sqrt{2}}{2}+ \imath$ and $\overline \lambda=\frac{\sqrt{2}}{2}-\imath$ and
the corresponding zero directions are 
$$  Q\Phi(\lambda) Q^\mathsf{H}  \begin{bmatrix} 1 \\ 0 \end{bmatrix}=0 \quad \Longleftrightarrow \quad \Phi(\lambda) \begin{bmatrix} 1 \\ \imath \end{bmatrix} /\sqrt{2} =0, 
$$
and
$$
Q\Phi\left(\overline \lambda\right) Q^\mathsf{H} \begin{bmatrix} 0 \\ 1 \end{bmatrix} =0 \quad \Longleftrightarrow \quad \Phi\left(\overline \lambda\right)\begin{bmatrix} 1 \\ -\imath \end{bmatrix}/\sqrt{2}=0.
$$
The interpolation conditions then are
\begin{align*}
  Z(\lambda) Q^\mathsf{H}\begin{bmatrix} 1 \\ 0 \end{bmatrix} &=
Z(\lambda) \begin{bmatrix} \frac{1}{\sqrt{2}} \\ \frac{\imath}{\sqrt{2}} \end{bmatrix} = W, ~~ \mathrm{and}~~
Z(\overline \lambda) \begin{bmatrix} \frac{1}{\sqrt{2}} \\ -\frac{\imath}{\sqrt{2}} \end{bmatrix} = \overline W,
\end{align*}
where $QZ(\lambda)Q^\mathsf{H}\begin{bmatrix} 1 \\ 0 \end{bmatrix}= 
\begin{bmatrix} 2-\frac{1}{\lambda -(a+\imath b)} \\ 0 \end{bmatrix} 
= \begin{bmatrix} \sqrt{2} \\ 0 \end{bmatrix}$ implies  
$$QW=\begin{bmatrix} \sqrt{2} \\ 0 \end{bmatrix},~~\text{and}~~W=\begin{bmatrix} 1 \\ \imath \end{bmatrix}.$$

The Loewner matrix then is obtained from $R=Q^\mathsf{H}$, $W=\sqrt{2}Q^\mathsf{H}$ and hence $W^\mathsf{H}R=\sqrt{2}I_2$, which finally yields 
\begin{align*}
 \mathbb{L}&=\frac{2\sqrt{2}}{\lambda+\overline \lambda}I_2= 2I_2,
 \quad \mathrm{}\\
 \mathbb{L}_\sigma &=
 \begin{bmatrix} \sqrt{2}\frac{\lambda-\overline \lambda}{\lambda+\overline \lambda}  & 0 \\ 0 & -\sqrt{2}\frac{\lambda-\overline \lambda}{\lambda+\overline \lambda}  \end{bmatrix}
 =2 \begin{bmatrix} \imath  & 0 \\ 0 & -\imath  \end{bmatrix},
 \end{align*}
 and the generalized state space realization \eqref{Loewner} becomes
 $$  
\left[ \begin{array}{c|c}  A -s E &  B \\ \hline  C &  D \end{array}\right] = 
\left[ \begin{array}{c|c} \mathbb{L}_{\sigma} - s  \mathbb{L}  & -\sqrt{2}Q \\ \hline -\sqrt{2}Q^\mathsf{H} & 0 \end{array}\right] +
 2 \left[ \begin{array}{c}- Q \\ \hline I_2 \end{array}\right] \left[ \begin{array}{c|c} Q^\mathsf{H} & I_2 \end{array}\right].
 $$
 Using the factorization $\mathbb{L}=\Gamma^\mathsf{H}\Gamma$ with $\Gamma:=\sqrt{2}Q^\mathsf{H}$, we get $\mathbb{L}_{\sigma \Gamma}=\frac12 Q^\mathsf{H}\mathbb{L}_{\sigma}Q$ and
 \begin{align*}
\left[ \begin{array}{c|c}  A_\Gamma  &  B_\Gamma \\ \hline  C_\Gamma &  D \end{array}\right] &= 
\left[ \begin{array}{c|c} \mathbb{L}_{\sigma \Gamma}  & - I_2 \\ \hline - I_2 & 0 \end{array}\right] +
 2\left[\begin{array}{c} - I_2/\sqrt{2} \\ \hline I_2 \end{array}\right] \left[ \begin{array}{c|c} I_2/\sqrt{2} & I_2 \end{array}\right]\\
 &=\left[ \begin{array}{cc|cc}  -1  & 1 &  -c & 0  \\ -1 & -1 & 0 & -c \\ \hline  1/c  &  0 & 2 & 0 \\  0 & 1/c & 0 & 2 \end{array}\right] 
\end{align*}
 with $c=1+\sqrt{2}$ and $1/c=\sqrt{2}-1$. The smallest eigenvalue $\lambda_{\min}{\left(\mathcal W_\Gamma\right)}$ of the above model is $0$ which is a poor estimate of its passivity radius.
 But we can apply to this model a similarity scaling with $T=cI_2$, which yields a model $\M_T$ where $A_T=A_\Gamma$ and $D_T=D$ are unchanged but $C_T=-B_T=I_2$. This
 corresponds to using \cite[Lemma 3.2]{MeV19} with the certificate $X=c^{-2}I$, and transforming the model to a new port-Hamiltonian 
 form  which has a passivity radius equal to $\lambda_{\min}{\left(\mathcal W_T\right)}=\frac12(3-\sqrt{5})\approx 0.382$.

\subsection{Electric RLC circuit}
As second example, we discuss the electrical circuit example considered in \cite{morAnt05a}. The system dynamics in the state-space form is given by as follows:
\begin{equation*}
\begin{aligned}
\dot x(t) &= A x(t) + B u(t),\\
y(t)&= C x(t) + Du(t),
\end{aligned}
\end{equation*}
where 
\begin{equation*}
A = \begin{bmatrix} -20 & -10& 0 & 0 & 0 \\ 10 &0 & -10 & 0 & 0 \\ 0&10 &0 & -10 & 0\\  0&0&10 &0 & -10  \\  0&0&0 &10 & -2 \end{bmatrix},~ B = \begin{bmatrix} 20\\0\\0\\0\\0 \end{bmatrix}, ~~ C =\begin{bmatrix}-2\\0\\0\\0\\0\end{bmatrix}^{\mathsf{T}},
\end{equation*}
and $D = 2$.

To identify the dynamics, we assume to have  $20$ points on the imaginary axis in a log-scale between $[10^{-1},10^3]$. We first employ the Loewner approach \cite{MaA07} to obtain a realization. In \Cref{fig:Ant_RLC_SVD}, we plot the singular values of the Loewner matrix, which allows us to determine the order of a minimal realization. We observe that the singular values after the $5$th are at the level of machine precision as one would expect. Hence, we determine a realization of order $5$. Next, we show the spectral zeros of the original and Loewner model in \Cref{fig:Antu_RLC_SpectralZero}, indicating that the spectral zeros of  both models are the same as expected. 

\begin{figure}[!tb]
	\centering
	\setlength\fheight{4.0cm}
	\setlength\fwidth{.4\textwidth}
	{\small	%
%
%
\definecolor{mycolor1}{rgb}{0.00000,0.44700,0.74100}%
\begin{tikzpicture}

\begin{axis}[%
width=\fwidth,
height=\fheight,
scale only axis,
xmin=0,
xmax=20,
ymode=log,
ymin=1e-20,
ymax=100000,
yminorticks=true,
grid = major,
ylabel = {$\sigma_k$},
xlabel = {$k$},
axis background/.style={fill=white}
]
\addplot [color=mycolor1, line width=2.0pt, forget plot]
  table[row sep=crcr]{%
1	2.25389268364263\\
2	0.566349658731192\\
3	0.369583691755613\\
4	0.06338898528689\\
5	0.0525023643049368\\
6	4.07624988621042e-15\\
7	3.58065313716367e-15\\
8	1.52471865047902e-15\\
9	1.25877636361754e-15\\
10	1.10333989052588e-15\\
11	3.43258553362484e-16\\
12	1.7975446204801e-16\\
13	1.23694599284471e-17\\
14	4.2050843378428e-18\\
15	2.40419274815928e-18\\
16	1.19019356772664e-18\\
17	1.18163493301964e-18\\
18	6.79869144243949e-19\\
19	2.82998405270147e-19\\
20	1.41028710198163e-20\\
};
\end{axis}
\end{tikzpicture}%
%
 }
	\caption{RLC example: The decay of the singular values of the Loewner matrix.}
	\label{fig:Ant_RLC_SVD}
\end{figure}
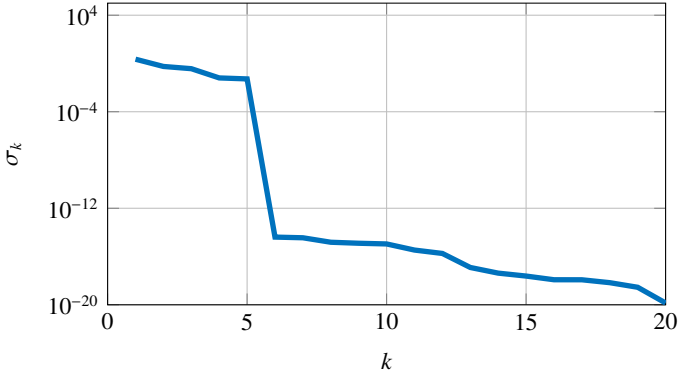

\begin{figure}[!tb]
    \centering
    \hspace{1.5cm} 
    \begin{tikzpicture}
      \begin{customlegend}[legend columns=2, legend style={/tikz/every even column/.append style={column sep=0.50cm}} , legend entries={~Orig. model, ~Loewner model}, ]
	\addlegendimage{color=blue,only marks,mark = +, mark size = 4,line width = 1.2pt}
	\addlegendimage{color=magenta,only marks, mark = o, mark size = 4,line width = 1.2pt}
      \end{customlegend}
    \end{tikzpicture}
    \setlength\fheight{4.0cm}
    \setlength\fwidth{5.0cm}
    {\small%
%
%
\begin{tikzpicture}

\begin{axis}[%
width=\fwidth,
height=\fheight,
at={(0\fwidth,0\fheight)},
scale only axis,
xmin=-3,
xmax=3,
xlabel style={font=\color{white!15!black}},
xlabel={Real part},
ymin=-20,
ymax=20,
ylabel style={font=\color{white!15!black}},
ylabel={Imag part},
axis background/.style={fill=white},
title style={font=\bfseries},
legend style={legend cell align=left, align=left, draw=white!15!black, at = {(1,1.2)}}
]
\addplot [color=blue, draw=none, mark=+, only marks,mark options={solid, blue}, mark size = 5,line width = 1.2pt]
  table[row sep=crcr]{%
-2.11289858183925	0\\
-1.59259838538348	10.07255612263\\
-1.59259838538348	-10.07255612263\\
-0.536178929334537	17.366624335962\\
-0.536178929334537	-17.366624335962\\
0.536178929334553	17.366624335962\\
0.536178929334553	-17.366624335962\\
1.59259838538347	10.07255612263\\
1.59259838538347	-10.07255612263\\
2.11289858183926	0\\
};

\addplot [color=magenta, draw=none, mark=o, mark options={solid, magenta},only marks, mark size = 5,line width = 1.2pt]
  table[row sep=crcr]{%
-2.11289858183923	0\\
-1.59259838538347	10.07255612263\\
-1.59259838538347	-10.07255612263\\
-0.536178929334554	17.366624335962\\
-0.536178929334554	-17.366624335962\\
0.536178929334539	17.366624335962\\
0.536178929334539	-17.366624335962\\
1.59259838538349	10.07255612263\\
1.59259838538349	-10.07255612263\\
2.11289858183923	0\\
};

\end{axis}

\end{tikzpicture}%
%
 }
    \caption{RLC example: Spectral zero of the original and Loewner model.}
    \label{fig:Antu_RLC_SpectralZero}
\end{figure}
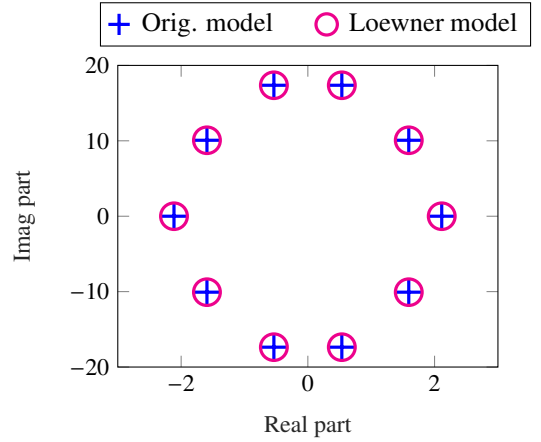

It is not in the form of a passive port-Hamiltonian system. But we can use the spectral zeros and zero directions of the Loewner model, which in this case, are the same as the original system, and estimate the transfer function at the spectral zeros along with the respective zero directions. Consequently, we  apply \Cref{Algo:Construct_pH_realization} to obtain a realization in the generalized state-space form of a port-Hamiltonian system \eqref{gpH} where 

\begin{align*}
 Q^{-1} &= \begin{bmatrix}  
    0.8795  &  0.0263  & -0.0304 &  -0.0511 &   0.0938\\
    0.0263  &  0.8515  & -0.0770 &  -0.1574 &  -0.0098\\
   -0.0304  & -0.0770  &  0.2545 &   0.0814 &   0.1136\\
   -0.0511  & -0.1574  &  0.0814 &   0.3560 &   0.0400\\
    0.0938  & -0.0098  &  0.1136 &   0.0400 &   0.2891\end{bmatrix},\\\displaybreak[3]
J &= \begin{bmatrix}   
         0 & -15.2595 &   0.5921 &   1.7823 &   0.5344\\
   15.2595 &        0 &  -0.4864 &  -0.8033 &   1.6342\\
   -0.5921 &   0.4864 &        0 &   0.5204 &  -0.5325\\
   -1.7823 &   0.8033 &  -0.5204 &        0 &  -3.3854\\
   -0.5344 &  -1.6342 &   0.5325 &   3.3854 &        0
   \end{bmatrix},\\
R &= \begin{bmatrix}    
    4.0000  &  0.0000 &  -2.8284 &  -3.9606 &   0.5598\\
    0.0000  &       0 &  -0.0000 &  -0.0000 &   0.0000\\
   -2.8284  & -0.0000 &   2.0000 &   2.8006 &  -0.3959\\
   -3.9606  & -0.0000 &   2.8006 &   3.9216 &  -0.5543\\
    0.5598  &  0.0000 &  -0.3959 & -0.5543  &  0.0784
\end{bmatrix},\\  
G & = \begin{bmatrix}-0.6563  &  0.3238 &   0.5378 &   0.6924 &  -0.2925\end{bmatrix}^{\mathsf{T}},\\
P & = \begin{bmatrix} 2.8284  &  0.0000 &  -2.0000 &  -2.8006 &   0.3959 \end{bmatrix}^{\mathsf{T}},\\
N & = 0,\quad S  = 2.
\end{align*}

Furthermore, we compare the Bode plots of the original, Loewner, and port-Hamiltonian model \eqref{fig:RLC_TF_Loew}, illustrating that all three models have the same transfer functions, and also have the same spectral zeros and zero directions. 

\begin{figure}[!tb]
    \centering
    \hspace{.1cm}
    \begin{tikzpicture}
	\begin{customlegend}[legend columns=2, legend style={/tikz/every even column/.append style={column sep=0.0cm}} , legend entries={Orig. model, Loewner model, port-Hamiltonian model}, ]
	\addlegendimage{color=blue, line width = 1.2pt}
	\addlegendimage{color=magenta,dashed,mark options={solid, magenta}, mark = +, mark size = 3, line width = 1.2pt}
	\addlegendimage{color=cyan,dashdotted,mark options={solid, cyan}, mark = o, mark size = 3, line width = 1.2pt}
	\end{customlegend}
    \end{tikzpicture}
    \setlength\fheight{5.0cm}
    \setlength\fwidth{.4\textwidth}
    {\small  %
	\input{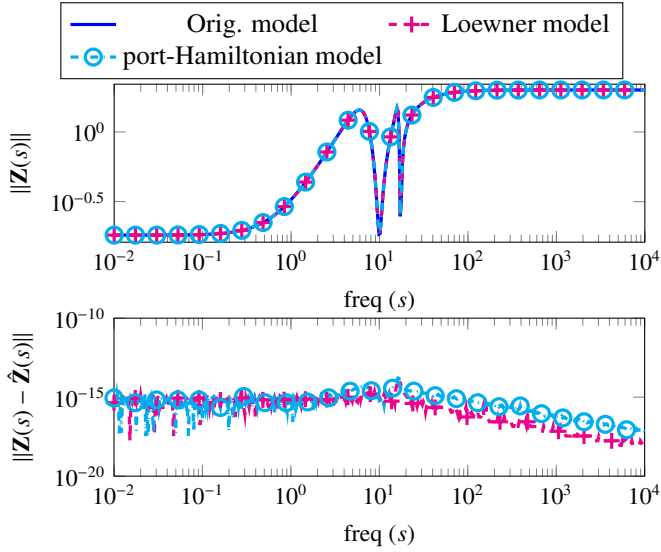}%
  }
    \caption{RLC example: Comparison of the Bode plots of the original, Loewner and port-Hamiltonian models.}
    \label{fig:RLC_TF_Loew}
\end{figure}

\subsection{A large scale electrical circuit}
Next, we consider a large scale RLC circuit, where $100$ electrical capacitances, inductors, and resistances are interconnected. For more details on the circuit topology, we refer to \cite{morGugA03}. The modeling of such a circuit leads to a model of order $n = 200$. Next, we assume that we obtain $200$ points on the imaginary axis on a log-scale within the range $\begin{bmatrix}10^{-1},10^3 \end{bmatrix}$. 

\begin{figure}[b]
	\centering
	\setlength\fheight{4.0cm}
	\setlength\fwidth{.4\textwidth}
	{\small%
	\input{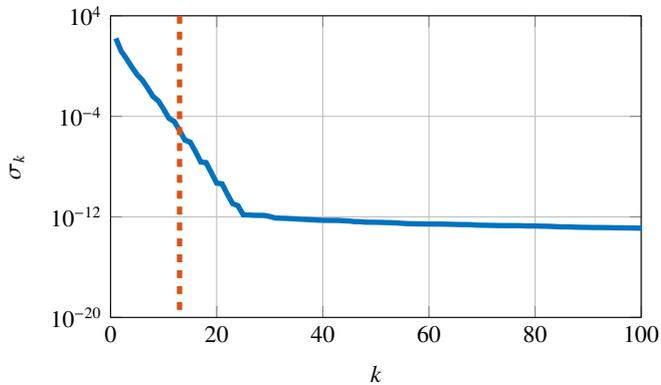}%
	}
	\caption{Large-scale RLC circuit: The decay of the singular values of the Loewner matrix.}
	\label{fig:LargescaleRLC_SVD}
\end{figure}

Towards constructing a port-Hamiltonian reduced-order system using the data, we first determine a realization using the classical Loewner method. We plot the decay of the singular values of the Loewner matrix in \Cref{fig:LargescaleRLC_SVD}, indicating a sharp decay. Having truncated singular values at $10^{-8}$   (relatively), we construct a reduced-order (Loewner) model of order $r = 13$, which is expected to capture the dynamics very well. Next, we compare the spectral zeros of the original and Loewner models in \Cref{fig:largescaleRLC_SpectralZero}. It is interesting to see that spectral zeros of both models are very different. Somehow, one can think of representative spectral zeros of the original systems with a smaller number, yet capturing the dynamics of the original systems very accurately.

\begin{figure}[!tb]
	\centering
    \hspace{1.5cm}
    \begin{tikzpicture}
	\begin{customlegend}[legend columns=2, legend style={/tikz/every even column/.append style={column sep=0.5cm}} , legend entries={~Orig. model, ~Loewner model}, ]
	\addlegendimage{color=blue,dashed,mark options={solid, blue}, mark = asterisk, mark size = 3, line width = 1.2pt, only marks}
	\addlegendimage{color=magenta,dashed,mark options={solid, magenta}, mark = o, mark size = 3, line width = 1.2pt,  only marks}
	\end{customlegend}
    \end{tikzpicture}
	\setlength\fheight{4.0cm}
	\setlength\fwidth{5.0cm}
	{\small		%
	\input{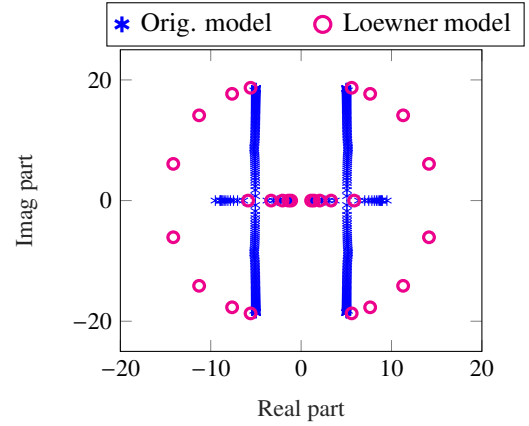}%
}
	\caption{A large scale RLC example: The decay of the singular values of the Loewner matrix.}
	\label{fig:largescaleRLC_SpectralZero}
\end{figure}
Subsequently, we determine the spectral zeros and zero directions using the Loewner model and evaluate the transfer function of the Loewner model at these zeros along with the respective directions. Then, we can determine a port-Hamiltonian realization using \Cref{Algo:Construct_pH_realization}. To compare the quality of models, we plot the Bode plots of the original, Loewner, and port-Hamiltonian models in \Cref{fig:largescaleRLC_TF_Loew}, showing the Loewner and port-Hamiltonian models approximate the original model very well. 

\begin{figure}[!tb]
	\centering
	    \hspace{.1cm}
    \begin{tikzpicture}
	\begin{customlegend}[legend columns=2, legend style={/tikz/every even column/.append style={column sep=0.0cm}} , legend entries={Orig. model, Loewner model, port-Hamiltonian model}, ]
	\addlegendimage{color=black, line width = 1.2pt}
	\addlegendimage{color=magenta,dashed,mark options={solid, magenta}, mark = +, mark size = 3, line width = 1.2pt}
	\addlegendimage{color=cyan,dashdotted,mark options={solid, cyan}, mark = o, mark size = 3, line width = 1.2pt}
	\end{customlegend}
    \end{tikzpicture}
	\setlength\fheight{5.0cm}
	\setlength\fwidth{.4\textwidth}
	{\small
		%
	\input{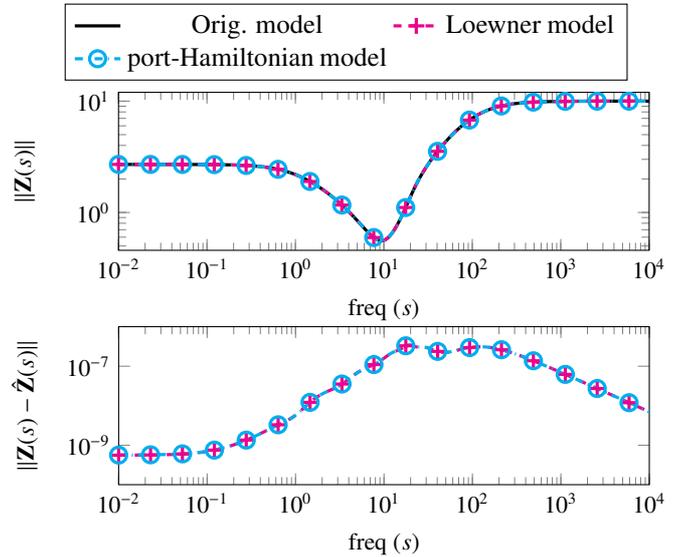}%

	}
	\caption{A large-scale RLC example: Comparison of the Bode plots of the original and Loewner model.}
	\label{fig:largescaleRLC_TF_Loew}
\end{figure}

\subsection{Frequency-limited port-Hamiltonian realization}
Lastly, we discuss the construction of a frequency-limited port-Hamiltonian realization using the same example as in the previous subsection. This means that the transfer function of the  port-Hamiltonian realization is required to very accurate in a given frequency band. Let us assume that we are given measurements in a frequency band $\begin{bmatrix} 5,15\end{bmatrix}$. As the first, we construct a Loewner model. This is followed by determining a reduced-order model of order $r = 9$. Next, we compare the spectral zeros of original and Loewner models in \Cref{fig:RLC_SpectralZeroFL}. It can be observed that the spectral zeros not only different from the original ones but also from those of the reduced-order model of order $r = 13$ in the previous example, see \Cref{fig:largescaleRLC_SpectralZero}. 
\begin{figure}[!tb]
	\centering
    \hspace{1.0cm}
    \begin{tikzpicture}
	\begin{customlegend}[legend columns=-1, legend style={/tikz/every even column/.append style={column sep=0.5cm}} , legend entries={~Orig. model, ~Loewner model (freq limit)},]
	\addlegendimage{color=blue,dashed,mark options={solid, blue}, mark = asterisk, mark size = 4, line width = 1.2pt, only marks}
	\addlegendimage{color=magenta,dashed,mark options={solid, magenta}, mark = o, mark size = 4, line width = 1.2pt,  only marks}
	\end{customlegend}
    \end{tikzpicture}
	\setlength\fheight{4.0cm}
	\setlength\fwidth{5.0cm}
	{\small		%
	\input{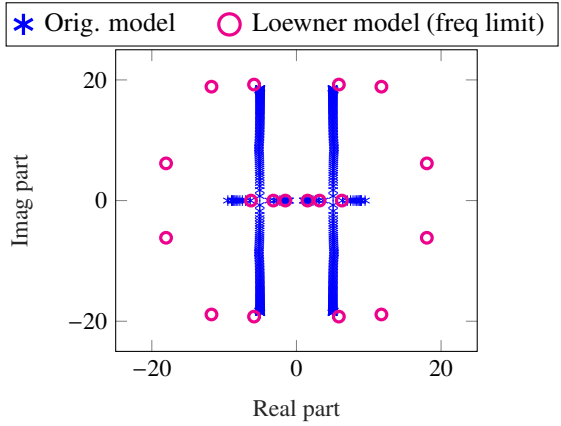}%
}
	\caption{Frequency limited RLC circuit: Comparison of spectral zeros of the original and Loewner model.}
	\label{fig:RLC_SpectralZeroFL}
\end{figure}

Next, we plot the transfer functions of the original and the identified port-Hamiltonian realization in \Cref{fig:RLC_TF_LoewFL}. Comparing, in particular, the error plots in \Cref{fig:largescaleRLC_TF_Loew,fig:RLC_TF_LoewFL}, we observe that the identified port-Hamiltonian realization using the data in the frequency band is much more accurate in the considered frequency band (nearly by three orders of magnitude) than the model identified in the previous subsection, and moreover, it is of a lower dimension.

\begin{figure}[!tb]
	\centering
	    \hspace{.1cm}
    \begin{tikzpicture}
	\begin{customlegend}[legend columns=2, legend style={/tikz/every even column/.append style={column sep=0.0cm}} , legend entries={Orig. model, Loewner model, port-Hamiltonian model}, ]
	\addlegendimage{color=black, line width = 1.2pt}
	\addlegendimage{color=magenta,dashed,mark options={solid, magenta}, mark = +, mark size = 4, line width = 1.2pt}
	\addlegendimage{color=cyan,dashdotted,mark options={solid, cyan}, mark = o, mark size =4, line width = 1.2pt}
	\end{customlegend}
    \end{tikzpicture}
	\setlength\fheight{5.0cm}
	\setlength\fwidth{.4\textwidth}
	{\small
		%
	\input{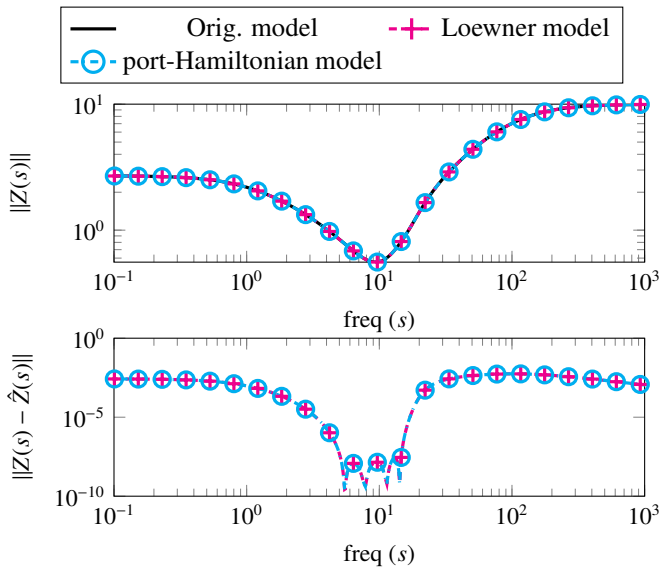}%

	}
	\caption{Frequency limited RLC circuit: Comparison of the Bode plots of the original and Loewner model.}
	\label{fig:RLC_TF_LoewFL}
\end{figure}

\section{Conclusions}
In this work, we have studied the identification problem for strictly passive systems. We have proposed a variant of the classical Loewner approach \cite{MaA07}, which constructs a realization in port-Hamiltonian form. We have also discussed a two-step procedure which allows us to construct a port-Hamiltonian realization using data on the imaginary axis. Furthermore, we have investigated the construction of frequency-limited port-Hamiltonian realization, which can also be viewed as a frequency-limited model-order reduction scheme for passive systems. We have illustrated the proposed methods by means of a couple of variants of electrical circuits. As a future direction, it would be interesting to investigate an identification problem of second-order passive systems by extending the idea proposed in \cite{duff2019data}.
\bibliographystyle{elsarticle-num}
\bibliography{mor}

\begin{thebibliography}{10}
\expandafter\ifx\csname url\endcsname\relax
  \def\url#1{\texttt{#1}}\fi
\expandafter\ifx\csname urlprefix\endcsname\relax\def\urlprefix{URL }\fi
\expandafter\ifx\csname href\endcsname\relax
  \def\href#1#2{#2} \def\path#1{#1}\fi

\bibitem{morBenQQ04a}
P.~Benner, E.~S. Quintana-Ort\'{\i}, G.~Quintana-Ort\'{\i}, Computing passive
  reduced-order models for circuit simulation, in: Proc. Intl. Conf. Parallel
  Comp. in Elec. Engrg. PARELEC 2004, IEEE Computer Society, Los Alamitos, CA,
  2004, pp. 146--151.

\bibitem{morDanP02}
L.~Daniel, J.~Phillips, Model order reduction for strictly passive and causal
  distributed systems, in: Proc. 39rd ACM IEEE Design Automation Conference,
  2002.

\bibitem{morFreF96a}
R.~W. Freund, P.~Feldmann, Reduced-order modeling of large passive linear
  circuits by means of the {SyPVL} algorithm, in: Technical Digest of the 1996
  IEEE/ACM International Conference on Computer-Aided Design, IEEE Computer
  Society Press, 1996, pp. 280--287.

\bibitem{morPhiDS03}
J.~R. Phillips, L.~Daniel, M.~Silveira, Guaranteed passive balancing
  transformations for model order reduction, IEEE Trans. Comput.-Aided Design
  Integr. Circuits Syst. 22~(8) (2003) 1027--1041.

\bibitem{gugercin2012structure}
S.~Gugercin, R.~V. Polyuga, C.~Beattie, A.~Van Der~Schaft, Structure-preserving
  tangential interpolation for model reduction of port-{H}amiltonian systems,
  Automatica 48~(9) (2012) 1963--1974.

\bibitem{polyuga2012effort}
R.~V. Polyuga, A.~J. van~der Schaft, Effort-and flow-constraint reduction
  methods for structure preserving model reduction of port-{H}amiltonian
  systems, Systems Control Lett. 61~(3) (2012) 412--421.

\bibitem{wolf2010passivity}
T.~Wolf, B.~Lohmann, R.~Eid, P.~Kotyczka, Passivity and structure preserving
  order reduction of linear port-{H}amiltonian systems using krylov subspaces,
  Eur. J. Control 16~(4) (2010) 401--406.

\bibitem{Wil72b}
J.~C. Willems, Dissipative dynamical systems part {II}: {L}inear systems with
  quadratic supply rates, Arch. Ration. Mech. Anal. 45~(5) (1972) 352--393.

\bibitem{BMX16}
C.~Beattie, V.~Mehrmann, H.~Xu, Port-{H}amiltonian realizations of linear time
  invariant systems, Tech. rep., TU Berlin (2016).

\bibitem{MeV19}
V.~Mehrmann, P.~Van~Dooren, Optimal robustness of port-{H}amiltonian systems,
  Tech. rep., TU Berlin (2019).

\bibitem{HaK75}
M.~Hazewinkel, R.~E. Kalman, On invariants, canonical forms and moduli for
  linear, constant, finite dimensional, dynamical systems, in: Mathematical
  Systems Theory, Springer, 1976, pp. 48--60.

\bibitem{Ant05}
A.~C. Antoulas, Approximation of Large-Scale Dynamical Systems, {SIAM}
  Publications, Philadelphia, PA, 2005.

\bibitem{MaA07}
A.~J. Mayo, A.~C. Antoulas, A framework for the solution of the generalized
  realization problem, Linear Algebra Appl. 425~(2-3) (2007) 634--662.

\bibitem{YoS67}
D.~C. Youla, M.~Saito, Interpolation with positive real functions, Journal of
  the Franklin Institute 284~(2) (1967) 77--108.

\bibitem{morAnt05a}
A.~C. Antoulas, A new result on passivity preserving model reduction, Systems
  Control Lett. 54~(4) (2005) 361--374.

\bibitem{morGugA03}
S.~Gugercin, A.~C. Antoulas, A survey of balancing methods for model reduction,
  in: Proc. European Control Conf. ECC~2003, Cambridge, UK, 2003, {CD Rom}.

\bibitem{duff2019data}
I.~P. Duff, P.~Goyal, P.~Benner, Data-driven identification of
  {R}ayleigh-damped second-order systems, {arXiv}:1910.00838 (2019).

\end{thebibliography}

\end{document}